\documentclass[11pt]{article}
\usepackage{algorithm2e}
\usepackage{framed}
\usepackage{amssymb}
\usepackage{amsfonts}
\usepackage{amsmath}
\usepackage{graphicx}
\usepackage{subfigure}
\usepackage{url}

\newtheorem{lemma}{Lemma}
\newtheorem{theorem}{Theorem}

\newenvironment{proof}{\noindent {\em Proof:}}{\\\hspace*{\fill}\mbox{$\diamond$}}
\newcommand{\argmin}{\text{argmin}}

\usepackage{nicefrac}

\newcommand{\nfrac}{\nicefrac}
\newcommand{\SDP}{{\sf SDP}\xspace}
\newcommand{\SPECTRAL}{{\sf SPECTRAL}\xspace}
\newcommand{\Tr}{{\rm Tr}}
\newcommand{\grad}{\nabla}
\newcommand{\cL}{\mathcal L}

\begin{document}


\title{Implementing regularization implicitly via \\ 
       approximate eigenvector computation}

\author{
Michael W. Mahoney
\thanks{
Department of Mathematics,
Stanford University,
Stanford, CA 94305,
{\tt mmahoney@cs.stanford.edu}.
}
\and
Lorenzo Orecchia
\thanks{
Computer Science Division, 
UC Berkeley,
Berkeley, CA, 94720
{\tt orecchia@eecs.berkeley.edu}.
}
}

\date{}
\maketitle

\begin{abstract}
Regularization is a powerful technique for extracting useful information 
from noisy data.
Typically, it is implemented by adding some sort of norm constraint to an 
objective function and then exactly optimizing the modified objective 
function.
This procedure often leads to optimization problems that are 
computationally more expensive than the original problem, a fact that is 
clearly problematic if one is interested in large-scale applications.
On the other hand, a large body of empirical work has demonstrated that 
heuristics, and in some cases approximation algorithms, developed to speed 
up computations sometimes have the side-effect of performing 
regularization implicitly.
Thus, we consider the question: What is the regularized optimization 
objective that an approximation algorithm is exactly optimizing?

We address this question in the context of computing approximations to the 
smallest nontrivial eigenvector of a graph Laplacian; and we consider 
three random-walk-based procedures: one based on the heat kernel of the 
graph, one based on computing the the PageRank vector associated with the 
graph, and one based on a truncated lazy random walk.
In each case, we provide a precise characterization of the manner in which 
the approximation method can be viewed as implicitly computing the exact 
solution to a regularized problem.
Interestingly, the regularization is not on the usual vector form of the 
optimization problem, but instead it is on a related semidefinite program.
\end{abstract}



\section{Introduction}
\label{sxn:intro}

Regularization is a powerful technique in statistics, machine learning, 
and data analysis for learning from or extracting useful information from
noisy data~\cite{Neu98,CH02,BL06}.
It involves (explicitly or implicitly) making assumptions about the data 
in order to obtain a ``smoother'' or ``nicer'' solution to a problem of 
interest.
The technique originated in integral equation theory, where it was of 
interest to give meaningful solutions to ill-posed problems for which a 
solution did not exist~\cite{TikhonovArsenin77}.
More recently, it has achieved widespread use in statistical data analysis, 
where it is of interest to achieve solutions that generalize well to 
unseen data~\cite{hast-tibs-fried}.
For instance, much of the work in kernel-based and manifold-based machine
learning is based on regularization in Reproducing kernel Hilbert 
spaces~\cite{SS01-book}.

Typically, regularization is implemented via a two step process: first, 
add some sort of norm constraint to an objective function of interest; 
and then, exactly optimize the modified objective function.
For instance, 
one typically considers a loss function $f(x)$ that specifies an empirical 
penalty depending on both the data and a parameter vector $x$; and a 
regularization function $g(x)$ that encodes prior assumptions about the 
data and that provides capacity control on the vector~$x$.
Then, one must solve an optimization problem of the~form:
\begin{equation}
\label{eqn:reg-gen}
\hat{x} = \argmin_x f(x) + \lambda g(x)   .
\end{equation}
A general feature of regularization implemented in this manner is that, 
although one obtains solutions that are ``better'' (in some statistical 
sense) than the solution to the original problem, one must often solve a 
modified optimization problem that is ``worse'' (in the sense of being 
more computationally expensive) than than the original optimization 
problem.%
\footnote{Think of ridge regression or the $\ell_1$-regularized 
$\ell_2$-regression problem.
More generally, however, even assuming that $g(x)$ is convex, 
one obtains a linear program or convex program that must solved.}
Clearly, this algorithmic-statistical tradeoff is problematic if one is 
interested in large-scale 
applications. 

On the other hand, it is well-known amongst practitioners that certain 
heuristics that can be used to speed up computations can sometimes have 
the side-effect of performing smoothing or regularization implicitly.
For example, ``early stopping'' is often used when a learning model such 
as a neural network is trained by an iterative gradient descent algorithm;
and ``binning'' is often used to aggregate the data into bins, upon which
computations are performed.
As we will discuss below, we have also observed a similar phenomenon in 
the empirical analysis of very large social and information 
networks~\cite{LLM10_communities_CONF}.
In these applications, the size-scale of the networks renders prohibitive 
anything but very fast nearly-linear-time algorithms, but the sparsity and 
noise properties of the networks are sufficiently complex that there is 
a need to understand the statistical properties \emph{implicit} in these 
fast algorithms in order to draw meaningful domain-specific conclusions 
from their output.

Motivated by these observations, we are interested in understanding in 
greater detail the manner in which algorithms that have superior 
algorithmic and computational properties either do or do not also have 
superior statistical properties.
In particular, we would like to know:
\begin{itemize}
\item
To what extent can one formalize the idea that performing an approximate 
computation can \emph{implicitly} lead to more regular solutions?
\end{itemize}
Rather than addressing this question in full generality, in this paper 
we will address it in the context of computing the first nontrivial 
eigenvector of the graph Laplacian.
(Of course, even this special case is of interest since a large body of 
work in machine learning, data analysis, computer vision, and scientific 
computation makes use of this vector.)
Our main result is a characterization of this implicit regularization 
in the context of three random-walk-based procedures for computing an 
approximation to this eigenvector.
In particular:
\begin{itemize} 
\item
We consider three random-walk-based procedures---one based on the heat 
kernel of the graph, one based on computing the the PageRank vector 
associated with the graph, and one based on a truncated lazy random 
walk---for computing an approximation to the smallest nontrivial 
eigenvector of a graph Laplacian, and we show that these approximation 
procedures may be viewed as implicitly solving a regularized optimization 
problem exactly.
\end{itemize}
Interestingly, in order to achieve this identification, we need to relax 
the standard spectral optimization problem to a semidefinite program.
Thus, the variables that enter into the loss function and the 
regularization term are not unit vectors, as they are more typically in 
formulations such as Problem~(\ref{eqn:reg-gen}), but instead they are 
distributions over unit vectors.
This was somewhat unexpected, and the empirical implications of this 
remain to be explored.

Before proceeding, let us pause to gain an intuition of our results in a 
relatively simple setting. 
To do so, 
consider the so-called Power Iteration Method, which takes as input an 
$n \times n$ symmetric matrix $A$ and returns as output a number 
$\lambda$ (the eigenvalue) and a vector $v$ (the eigenvector) such that 
$Av = \lambda v$.%
\footnote{Our result for the truncated lazy random walk generalizes a 
special case of the Power Method.  Formalizing the regularization implicit 
in the Power Method more generally, or in other methods such as the 
Lanczos method or the Conjugate Gradient method, is technically more 
intricate due to the renormalization at each step, which by construction
we will not need.}
The Power Iteration Method starts with an initial random vector, call it 
$\nu_0$, and it iteratively computes $\nu_{t+1} = A \nu_t /||A\nu_t||_2$.
Under weak assumptions, the method converges to $v_1$, the dominant 
eigenvector of $A$.
The reason is clear: if we expand $\nu_0 = \sum_{i=1}^{n} \gamma_i v_i$ 
in the basis provided by the eigenfunctions $\{v_i\}_{i=1}^{n}$ of $A$, 
then $\nu_t = \sum_{i=1}^{n} \gamma_i^t v_i \rightarrow v_1$.
If we truncate this method after some very small number, say $3$, 
iterations, then the output vector is clearly a suboptimal approximation 
of the dominant eigen-direction of the particular matrix $A$; but due to 
the admixing of information from the other eigenvectors, it may be a 
better or more robust approximation to the best ``ground truth 
eigen-direction'' in the ensemble from which $A$ was drawn.
It is this intuition in the context of computing eigenvectors of the
graph Laplacian that our main results formalize.

\section{Overview of the problem and approximation procedures}
\label{sxn:background}

For a connected, weighted, undirected graph $G=(V,E)$, let $A$ be its 
adjacency matrix and $D$ its diagonal degree matrix, \emph{i.e.},
$D_{ii}=\sum_{j:(ij)\in E} w_{ij}$, where $w_{ij}$ is the weight of edge 
$(ij)$.
Let $M = AD^{-1}$ be the natural random walk transition matrix associated 
with $G$, in which case $W=(I+M)/2$ is the usual lazy random walk 
transition matrix. 
(Thus, we will be \emph{post}-multiplying by \emph{column} vectors.)
Finally, let $L = I - D^{-\nfrac{1}{2}}AD^{-\nfrac{1}{2}}$ be the 
normalized Laplacian of~$G$.

We start by considering the standard spectral optimization problem.
\begin{eqnarray*}
{\SPECTRAL:} & \min & x^T L x   \\
             & \mbox{s.t.} & x^T x = 1   \\
             & & x^T D^{1/2} 1 = 0   .
\end{eqnarray*}
In the remainder of the paper, we will assume that this last constraint 
always holds, effectively limiting ourselves to be in the subspace 
$\mathbb{R}^n \perp 1$, by which we mean 
$\{x\in\mathbb{R}^{n}:x^TD^{1/2}1=0\}$.
(Omitting explicit reference to this orthogonality constraint and 
assuming that we are always working in the subspace $\mathbb{R}^n \perp 1$ 
makes the statements and the proofs easier to follow and does not impact the 
correctness of the arguments.  To check this, notice that the proofs can be
carried out in the language of linear operators without any reference to a 
particular matrix representation in~$\mathbb{R}^{n}$.)

Next, we provide a description of three related random-walk-based matrices 
that arise naturally when considering a graph $G$.
\begin{itemize}
\item
\textbf{Heat Kernel.}
The Heat Kernel of a connected, undirected graph $G$ can be defined as:
\begin{equation}
\label{eqn:heat-kernel}
H_t = \exp ( -tL )  = \sum_{k=0}^{\infty} \frac{(-t)^k}{k!}L^k  ,
\end{equation}
where $t \ge 0$ is a time parameter.
Alternatively, it can be written as $H_t = \sum_i e^{-\lambda_it}P_i$, 
where $\lambda_i$ is the $i$-th eigenvalue of $L$ and $P_i$ denotes the 
projection into the eigenspace associated with $\lambda_i$.
The Heat Kernel is an operator that satisfies the heat equation 
$\frac{\partial H_t}{\partial t} = -LH_t$ and thus that describes the 
diffusive spreading of heat on the graph.
\item
\textbf{PageRank.}
The PageRank vector $\pi(\gamma,s)$ associated with a connected, 
undirected graph $G$ is defined to be the unique solution to 
$$
\pi(\gamma,s) = \gamma s + (1-\gamma) M \pi(\gamma,s)  ,
$$
where $\gamma \in (0,1)$ is the so-called teleportation constant; 
$s \in \mathbb{R}^{n}$ is a preference vector, often taken to be (up to 
normalization) the all-ones vector; and $M$ is the natural random walk 
matrix associated with $G$~\cite{LM04}.%
\footnote{Alternatively, one can define $\pi^{\prime}(\gamma,s)$ to be 
the unique solution to $\pi = \gamma s + (1-\gamma) W \pi $, where $W$ is 
the $\nfrac{1}{2}$-lazy random walk matrix associated with~$G$.  
These two vectors are related as
$\pi^{\prime}(\gamma,s)=\pi(\frac{2\gamma}{1+\gamma},s)$~\cite{andersen06local}. }
If we fix $\gamma$ and $s$, then it is known that 
$\pi(\gamma,s) = \gamma \sum_{t=0}^{\infty} (1-\gamma)^t M^t s$, and thus
that $\pi(\gamma,s) = R_{\gamma} s$, where
\begin{equation}
\label{eqn:page-rank}
R_{\gamma} = \gamma \left(I-\left(1-\gamma \right)M \right)^{-1}   .
\end{equation}
This provides an expression for the PageRank vector $\pi(\gamma,s)$ as a 
$\gamma$-dependent linear transformation matrix $R_{\gamma}$ multiplied 
by the preference vector $s$~\cite{andersen06local}.
That is, Eqn.~(\ref{eqn:page-rank}) simply states that PageRank can be 
presented as a linear operator $R_{\gamma}$ acting on the seed $s$.
\item
\textbf{Truncated Lazy Random Walk.}
Since $M = AD^{-1}$ is the natural random walk transition matrix 
associated with a connected, undirected graph $G$, it follows that
\begin{equation}
\label{eqn:lazy-walk}
W_{\alpha}= \alpha I + (1-\alpha)M
\end{equation}
represents one step of the $\alpha$-lazy random walk transition matrix, in 
which at each step there is a holding probability $\alpha \in [0,1]$.
Just as $M$ is similar to $M^{\prime}=D^{-1/2}MD^{1/2}$, which permits the 
computation of its real eigenvalues and full suite of eigenvectors that 
can be related to those of $M$, $W_{\alpha}$ is similar to 
$W_{\alpha}^{\prime}=D^{-1/2}W_{\alpha}D^{1/2}$.
Thus, iterating the random walk $W_{\alpha}$ is similar to applying the 
Power Method to $W_{\alpha}^{\prime}$, except that the renormalization at 
each step need not be performed since the top eigenvalue is unity.
\end{itemize}
Each of these three matrices has been used to compute vectors that in 
applications are then used in place of the smallest nontrivial eigenvector 
of a graph Laplacian.
This is typically achieved by starting with an initial random vector and 
then applying the Heat Kernel matrix, or the PageRank operator, or 
truncating a Lazy Random Walk.

Finally, we recall that the solution \SPECTRAL can also be characterized 
as the solution to a semidefinite program (SDP).
To see this, consider the following \SDP:
\begin{eqnarray*}
{\SDP:} &\min & L \bullet X  \\
        & \mbox{s.t.} & \Tr(X) = I \bullet X=1  \\
        & & X \succeq 0  ,
\end{eqnarray*}
where $\bullet$ stands for the Trace, or matrix inner product, operation, 
\emph{i.e.}, $A \bullet B = \Tr(A B^T)=\sum_{ij}A_{ij}B_{ij}$ for matrices 
$A$ and $B$.
(Recall that, both here and below, $I$ is the Identity on the subspace 
perpindicular to the all-ones vector.)
\SDP is a relaxation of the spectral program \SPECTRAL from an 
optimization over unit vectors to an optimization over distributions over 
unit vectors, represented by the density matrix $X$.

To see the relationship between the solution $x$ of \SPECTRAL and the 
solution $X$ of \SDP, recall that a density matrix $X$ is a matrix of second 
moments of a distribution over unit vectors.
In this case, $L \bullet X$ is the expected value of $x^TLx$, when $x$ is
drawn from a distribution defined by $X$.
If $X$ is rank-$1$, as is the case for the solution to $\SDP$, then the 
distribution is completely concentrated on a vector $v$, and the SDP and 
vector solutions are the same, in the sense that $X=vv^T$.
More generally, as we will encounter below, the solution to an SDP may not
be rank-$1$.
In that case, a simple way to construct a vector $x$ from a distribution
defined by $X$ is to start with an $n$-vector $\xi$ with entries drawn 
i.i.d. from the normal distribution $N(0,1/n)$, and consider $x=X^{1/2}\xi$.
Note that this procedure effectively samples from a Gaussian distribution
with second moment $X$.

\section{Approximation procedures and regularized spectral optimization problems}
\label{sxn:main}

\subsection{A simple theorem characterizing the solution to a regularized SDP}
\label{sxn:main-second}

Here, we will apply regularization technique to the SDP formulation
provided by \SDP, and we will show how natural regularization functions 
yield distributions over vectors which correspond to the diffusion-based 
or random-walk-based matrices.
In order to regularize \SDP, we want to modify it such that the 
distribution is not degenerate on the second eigenvector, but instead 
spreads the probability on a larger set of unit vectors around $v$.
The regularized version of \SDP we will consider will be of the form:
\begin{eqnarray*}
{\sf (F,\eta)-SDP} & \min & L  \bullet X + \nfrac{1}{\eta} \cdot F(X)   \\
                   & \mbox{s.t.} & I \bullet X=1   \\
                   & & X \succeq 0   ,
\end{eqnarray*}
where $\eta > 0$ is a trade-off or regularization parameter determining 
the relative importance of the regularization term $F(X)$, and where $F$ 
is a real strictly-convex infinitely-differentiable rotationally-invariant 
function over the positive semidefinite cone.
(Think of $F$ as a strictly convex function of the eigenvalues of $X$.) 
For example, $F$ could be the negative of the von Neumann entropy of $X$; 
this would penalize distributions that are too concentrated on a small 
measure of vectors.
We will consider other possibilities for $F$ below.
Note that due to $F$, the solution $X$ of ${\sf (F,\eta)-SDP}$ will in 
general not be rank-$1$. 

Our main results on implicit regularization via approximate computation 
will be based on the following structural theorem that provides sufficient 
conditions for a matrix to be a solution of a regularized SDP of a certain 
form.
Note that the Lagrangian parameter $\lambda$ and its relationship with
the regularization parameter $\eta$ will play a key role in relating this
structural theorem to the three random-walk-based proceudres described 
previously.

\begin{theorem}
\label{thm:main}
Let $G$ be a connected, weighted, undirected graph, with normalized 
Laplacian $L$.
Then, the following conditions are sufficient for $X^\star$ to be an 
optimal solution to ${\sf (F,\eta)-SDP}$.
\begin{enumerate}
\item 
$X^{\star} = (\grad F)^{-1} \left(\eta \cdot ( \lambda^{*} I - L)\right)$, for 
some $\lambda^{*} \in \mathbb{R}$,
\item $I \bullet X^{\star} = 1$,
\item $X^\star \succeq 0.$
\end{enumerate}
\end{theorem}
\begin{proof}
For a general function $F$, we can write the Lagrangian $\cL$ for 
${\sf (F,\eta)-SDP}$ as follows:
\begin{eqnarray*}
\cL(X, \lambda, U) = L \bullet X + \frac{1}{\eta} \cdot F(X) - \lambda \cdot (I \bullet X -1) - U \bullet X
\end{eqnarray*}
where $\lambda \in \mathbb{R}, U \succeq 0.$
The dual objective function is 
$$h(\lambda, U) = \min_{X \succeq 0} \cL(X, \lambda, U).$$
As $F$ is strictly convex, differentiable and rotationally invariant, the gradient of $F$ over the positive semidefinite cone is invertible and the righthand side is minimized when 
$$X=(\grad F)^{-1}(\eta \cdot ( -L +\lambda^{*} \cdot I + U)) , $$
where $\lambda^{*}$ is chosen such that the second condition in the statement
of the theorem is satisfied.
Hence, 
$$
h(\lambda^\star, 0) = L \bullet X^\star + \frac{1}{\eta} \cdot F(X^\star) - \lambda^\star \cdot (I \bullet X^\star -1 ) = L \bullet X^\star + \frac{1}{\eta} \cdot F(X^\star)  .
$$
By Weak Duality, this implies that $X^\star$ is an optimal solution to ${\sf (F,\eta)-SDP}.$
\end{proof}

\noindent
Two clarifying remarks regarding this theorem are in order.
First, the fact that such a $\lambda^*$ exists is an assumption of the 
theorem.
Thus, in fact, the theorem is just a statement of the KKT conditions and 
strong duality holds for our SDP formulations. 
For simplicity and to keep the exposition self-contained, we decided to 
present the proof of optimality, which is extremely easy in the case of 
an SDP with only linear constraints.
Second, we can plug the dual solution $(\lambda^*, 0)$ into the dual 
objective and show that, under the assumptions of the theorem, we obtain a 
value equal to the primal value of $X^*$. 
This certifies that $X^*$ is optimal. 
Thus, we do not need to assume $U=0$; we just choose to plug in this 
particular dual solution.

\subsection{The connection between approximate eigenvector computation and implicit statistical regularization}
\label{sxn:main-third}

In this section, we will consider the three diffusion-based or 
random-walked-based heuristics described in Section~\ref{sxn:background}, 
and we will show that each may be viewed as solving ${\sf (F,\eta)-SDP}$ 
for an appropriate value of $F$ and $\eta$.

\paragraph{Generalized Entropy and the Heat Kernel.} 
Consider first the Generalized Entropy function: 
\begin{equation}
\label{eqn:gen-ent}
F_H(X) = \Tr(X \log X) - \Tr(X)   , 
\end{equation}
for which:
\begin{eqnarray*}
(\grad F_H)(X)      &=& \log X \\
(\grad F_H)^{-1}(Y) &=& \exp{Y}   .
\end{eqnarray*}
Hence, the solution to $({\sf F_H, \eta})-\SDP$ has the form:
\begin{equation}
\label{eqn:Sol-gen-ent}
X_H^\star = \exp(\eta \cdot (\lambda I - L))    ,
\end{equation}
for appropriately-chosen values of $\lambda$ and $\eta$.
Thus, we can establish the following lemma.

\begin{lemma}
Let $X_H^\star$ be an optimal solution to ${\sf (F,\eta)-\SDP}$, when 
$F(\cdot)$ is the Generalized Entropy function, given by 
Equation~(\ref{eqn:gen-ent}).
Then 
$$
X_H^\star 
          = \frac{H_{\eta}}{\Tr\left[H_{\eta}\right]}  ,
$$
which corresponds to a ``scaled'' version of the Heat Kernel matrix with 
time parameter $t=\eta$. 
\end{lemma}
\begin{proof}
From Equation~(\ref{eqn:Sol-gen-ent}), it follows that
$ X_H^\star = \exp(- \eta \cdot L) \cdot \exp(\eta \cdot \lambda) $, and
thus by setting $\lambda = -\nfrac{1}{\eta}\log(\Tr(\exp(-\eta\cdot L)))$, 
we obtain 
 the expression for $X_H^\star$ given in the lemma.
Thus, $X_H^\star \succeq 0$ and $\Tr(X_H^\star) = 1$, and by 
Theorem~\ref{thm:main} the lemma follows.
\end{proof}

\noindent
Conversely, given a graph $G$ and time parameter $t$, the 
Heat Kernel of Equation~(\ref{eqn:heat-kernel}) can be characterized 
as the solution to the regularized $({\sf F_H, \eta})-\SDP$, with 
the regularization parameter $\eta=t$ (and for the value of the Lagrangian 
parameter $\lambda$ as specified in the proof).

\paragraph{Log-determinant and PageRank.} 
Next, consider the Log-determinant function: 
\begin{equation}
\label{eqn:log-det}
F_D(X) =  - \log\det(X)  ,
\end{equation}
for which:
\begin{eqnarray*}
(\grad F_D)(X)      &=& -X^{-1} \\
(\grad F_D)^{-1}(Y) &=& -Y^{-1}   .
\end{eqnarray*}
Hence, the solution to $({\sf F_D, \eta})-\SDP$ has the form:
\begin{equation}
\label{eqn:Sol-log-det}
X_D^\star = - (\eta \cdot (\lambda I - L))^{-1}   ,
\end{equation}
for appropriately-chosen values of $\lambda$ and $\eta$.
Thus, we can establish the following lemma.

\begin{lemma}
Let $X_D^\star$ be an optimal solution to ${\sf (F,\eta)-\SDP}$, when 
$F(\cdot)$ is the Log-determinant function, given by 
Equation~(\ref{eqn:log-det}).
Then 
$$
X_D^\star = \frac{D^{-1/2} R_\gamma D^{1/2}}{\Tr\left[ R_\gamma \right]}
$$
which corresponds to a ``scaled-and-streached'' version of the 
PageRank matrix $R_\gamma$ of Equation~(\ref{eqn:page-rank}) with 
teleportation parameter $\gamma$  depending on  $\eta$.
\end{lemma}
\begin{proof}
Recall that $L = I - D^{-\nfrac{1}{2}}AD^{-\nfrac{1}{2}}$. 
Since $ X_D^\star = \nfrac{1}{\eta} \cdot (L - \lambda I)^{-1} $,
by standard manipulations it follows that
$$
X_D^\star = \frac{1}{\eta}
            \left( (1-\lambda) I - D^{-1/2}AD^{-1/2} \right)^{-1}   .
$$
Thus, $X_D^\star \succeq 0$ if $\lambda \le 0$, and
 $X_D^\star \succ 0$ if $\lambda < 0$.
If we set $\gamma = \frac{\lambda}{\lambda-1}$
(which varies from $1$ to $0$, as $\lambda$ varies from $-\infty$ to $0$),
then it can be shown that
$$
X_D^\star = \frac{-1}{\eta\lambda}
            D^{-1/2}
            \gamma
            \left( I - (1-\gamma)AD^{-1} \right)^{-1}   
            D^{1/2}   .
$$
By requiring that $1=\Tr[X_D^\star]$, it follows that 
$$
\eta = (1-\gamma) \Tr\left[ \left( I-(1-\gamma)AD^{-1}\right)^{-1}\right]
$$
and thus that
$\eta\lambda = - \Tr\left[\gamma(I-(1-\gamma)AD^{-1})^{-1}\right]$.
Since $ R_\gamma = \gamma(I-(1-\gamma)AD^{-1})^{-1}$, the lemma follows.
\end{proof}

\noindent
Conversely, given a graph $G$ and teleportation parameter $\gamma$, the 
PageRank of Equation~(\ref{eqn:page-rank}) can be characterized 
as the solution to the regularized $({\sf F_D, \eta})-\SDP$, with 
the regularization parameter $\eta$ as specified in the proof.

\paragraph{Standard $p$-norm and Truncated Lazy Random Walks.} 
Finally, consider the Standard $p$-norm function, for $p > 1$: 
\begin{equation}
\label{eqn:pnorm}
F_p(X) = \frac{1}{p} ||X||_p^p = \frac{1}{p} \Tr(X^p)   ,
\end{equation}
for which:
\begin{eqnarray*}
(\grad F_p)(X)      &=& X^{p-1} \\
(\grad F_p)^{-1}(Y) &=& Y^{\nfrac{1}{(p-1)}} .  
\end{eqnarray*}
Hence, the solution to $({\sf F_p, \eta})-\SDP$ has the form
\begin{equation}
\label{eqn:Sol-pnorm}
X_p^\star = (\eta \cdot(\lambda I -L))^{q-1} ,
\end{equation}
where $q > 1$ is such that $\nfrac{1}{p}+\nfrac{1}{q} = 1$,
for appropriately-chosen values of $\lambda$ and $\eta$.
Thus, we can establish the following lemma.
\begin{lemma}
Let $X_p^\star$ be an optimal solution to ${\sf (F,\eta)-\SDP}$, when 
$F(\cdot)$ is the Standard $p$-norm function, for $p > 1$, given by 
Equation~(\ref{eqn:pnorm}).
Then 
$$
X_p^\star = \frac{D^{\frac{-(q-1)}{2}} W_{\alpha}^{q-1}D^{\frac{q-1}{2}} }
                 { \Tr\left[ W_{\alpha}^{q-1}\right] }
$$
which corresponds to a ``scaled-and-streached'' version of $q-1$ steps of 
the Truncated Lazy Random Walk matrix $W_{\alpha}$ of 
Equation~(\ref{eqn:lazy-walk}) with laziness parameter $\alpha$ depending on 
$\eta$.
\end{lemma}
\begin{proof}
Recall that $L = I - D^{-\nfrac{1}{2}}AD^{-\nfrac{1}{2}}$. 
Since $ X_p^\star = \eta^{q-1} \cdot (\lambda I -L)^{q-1} $, 
by standard manipulations it follows that
$$
X_p^\star = \eta^{q-1} \left( (\lambda-1)I + D^{-1/2}AD^{-1/2} \right)^{q-1}
$$
Thus, $X_p^\star \succeq 0$ if $\lambda \ge 1$, and
 $X_p^\star \succ 0$ if $\lambda > 1$.
If we set $\alpha= \frac{\lambda-1}{\lambda}$
(which varies from $0$ to $1$, as $\lambda$ varies from $1$ to $\infty$),
then it can be shown that
$$
X_p^\star = (\eta\lambda)^{q-1}
            D^{-(q-1)/2}
            \left( \alpha I - (1-\alpha)AD^{-1} \right)^{q-1}   
            D^{(q-1)/2}   .
$$
By requiring that $1=\Tr[X_p^\star]$, it follows that 
$$
\eta = (1-\alpha) \left\{ \Tr\left[ \left( \alpha I + (1-\alpha) AD^{-1} \right)^{q-1} \right] \right\}^{1-p}
$$
and thus that
$\eta\lambda = \left\{ \Tr\left[ \left( \alpha I + (1-\alpha) AD^{-1} \right)^{q-1} \right] \right\}^{1-p}$.
Since $ W_\alpha = \alpha I + (1-\alpha) AD^{-1} $, the lemma follows.
\end{proof}

\noindent
Conversely, given a graph $G$, a laziness parameter $\alpha$, and a number 
of steps $q^{\prime}=q-1$, the
Truncated Lazy Random Walk of Equation~(\ref{eqn:lazy-walk})
can be characterized 
as the solution to the regularized $({\sf F_p, \eta})-\SDP$, with 
the regularization parameter $\eta$ as specified in the proof.

\section{Discussion and Conclusion}

There is a large body of empirical and theoretical work with a broadly 
similar flavor to ours.
Here, we provide just a few citations that most informed our approach.
\begin{itemize}
\item
In machine learning, 
Belkin, Niyogi, and Sindhwan describe a geometrically-motivated framework 
within which semi-supervised learning algorithms can be 
constructed~\cite{BNS06};
Saul and Roweis (and many others, but less explicitly) observe that adding 
a regularization term to improve numerical properties also ``acts to 
penalize large weights that exploit correlations beyond some level of 
precision in the data sampling process''~\cite{SR03};
Rosasco, De~Vito, and Verri describe how a large class of regularization 
methods designed for solving ill-posed inverse problems gives rise to 
novel learning algorithms~\cite{rosasco05_spectral};
Zhang and Yu show that in boosting, early stopping (as opposed to waiting
for full convergence) leads to regularization and hence better 
prediction~\cite{ZY05};
Shi and Yu describe statistical aspects of binning in Gaussian kernel 
regularization~\cite{SY05}; and
Bishop observes that training with noise can be equivalent to Tikhonov 
regularization~\cite{Bis95}.
\item
In numerical linear algebra,
O'Leary, Stewart, and Vandergraft have described issues that arise in 
estimating the largest eigenvalue of a positive definite matrix with 
the power method~\cite{OSV79}; and
Parlett, Simon, and Stringer have described convergence issues that arise
when estimating the largest eigenvalue with an iterative 
method~\cite{PSS82}.
\item
In the theory of algorithms,
Spielman and Teng describe how to perform local graph partitioning using
truncated random walks~\cite{Spielman:2004};
Andersen, Chung, and Lang describe an improved local graph partitioning 
algorithm PageRank vectors~\cite{andersen06local}; and
Chung describes how to perform similar operations by using the heat kernel 
and viewing it as the so-called pagerank of a 
graph~\cite{Chung07_heatkernelPNAS}.
\item
In internet data analysis, 
Andersen and Lang use these methods to try to find communities in large 
networks~\cite{andersen06seed};
Leskovec, Lang, and Mahoney use these and other methods to show that there 
do not exist good large communities in large social and information 
networks~\cite{LLDM08_communities_CONF,LLM10_communities_CONF};
and Lu \emph{et al.} empirically evaluate implicit regularization 
constraints for improved online review quality prediction~\cite{LTNP10}.
\end{itemize}
None of this work, however, takes the approach we have adopted of asking: 
What is the regularized optimization objective that a heuristic or 
approximation algorithm is exactly optimizing?


We should note that one can interpret our main results from one of two 
alternate perspectives.
From the perspective of worst-case analysis, we provide a simple 
characterization of several related methods for approximating the smallest
nontrivial eigenvector of a graph Laplacian as solving a related 
optimization problem.
By adopting this view, it should perhaps be less surprising that these 
methods have Cheeger-like inequalities, with related algorithmic 
consequences, associated with 
them~\cite{Spielman:2004,andersen06local,Chung07_heatkernelPNAS,chung07_fourproofs}.
From a statistical perspective, one could imagine one method or another 
being more or less appropriate as a method to compute robust approximations
to the smallest nontrivial eigenvector of a graph Laplacian, depending on 
assumptions being made about the data.
By adopting this view, it should perhaps be less surprising that these 
methods have performed well at identifying structure in sparse and noisy 
networks~\cite{andersen06seed,LLDM08_communities_CONF,LLM10_communities_CONF,LTNP10}.

The particular results that motivated us to ask this question had to do 
with recent empirical work on characterizing the clustering and community 
structure in very large social and information 
networks~\cite{LLDM08_communities_CONF,LLM10_communities_CONF}.
As a part of that line of work, Leskovec, Lang, and Mahoney 
(LLM)~\cite{LLM10_communities_CONF} were interested in understanding the 
artifactual properties induced in output clusters as a function of 
different approximation algorithms for a given objective function (that 
formalized the community concept).
LLM observed a severe tradeoff between the objective function value and 
the ``niceness'' of the clusters returned by different approximation 
algorithms.
This phenomenon is analogous to the bias-variance tradeoff that is 
commonly-observed in statistics and machine learning, except that LLM did 
not perform any explicit regularization---instead, they observed this 
phenomenon as a function of different approximation algorithms to compute 
approximate solutions to the intractable graph partitioning problem.

Although we have focused in this paper simply on the problem of computing 
an eigenvector, one is typically interested in computing eigenvectors in 
order to perform some downstream data analysis or machine learning task.
For instance, one might be interested in characterizing the clustering 
properties of the data.  
Alternatively, the goal might be to perform classification or regression or 
ranking.  
It would, of course, be of interest to understand how the concept of 
\emph{implicit regularization via approximate computation} extends to the 
output of algorithms for these problems.
More generally, though, it would be of interest to understand how this 
concept of implicit regularization via approximate computation extends to 
intractable graph optimization problems (that are not obviously 
formulatable as vector space problems) that are more popular in computer 
science.
That is: What is the (perhaps implicitly regularized) optimization problem
that an approximation algorithm for an intractable optimization problem is 
implicitly optimizing?
Such graph problems arise in many applications, but the the formulation 
and solution of these graph problems tends to be quite different than that 
of matrix problems that are more popular in machine learning and 
statistics.
Recent empirical and theoretical evidence, however, clearly suggests that 
regularization will be fruitful in this more general setting.


 \end{document}